\documentclass[a4paper,12pt]{article}

\usepackage[T1]{fontenc}
\usepackage{makecell}
\usepackage{enumitem}
\usepackage{amsmath}
\usepackage{amsthm}
\usepackage{amssymb}
\usepackage{algorithm}
\usepackage{algpseudocode}
\usepackage{hyperref}
\usepackage{color,soul}
\usepackage{multirow}
\usepackage{amsmath}

\usepackage[authoryear]{natbib}
\newtheorem{Theorem}{Theorem}[section]
\newtheorem{Lemma}[Theorem]{Lemma}

\theoremstyle{definition}

\newtheorem{Remark}{Remark}[section]

\usepackage{theoremref}
\usepackage{graphicx}
\usepackage[margin=1in]{geometry}
\usepackage[dvipsnames]{xcolor}
\usepackage{amsthm,thmtools}%

\newcommand{\BibTeX}{B\kern-.05em{\sc i\kern-.025em b}\kern-.08em\TeX}

\newcommand{\fullversion}[1]{}

\makeatletter

\usepackage{cleveref}

\title{Fair Allocation with Money: \\What is Your Objective?}

\author{Noga Klein Elmalem, Rica Gonen, Erel Segal-Halevi}
\begin{document}
\maketitle

\begin{abstract}
When allocating indivisible items, there are various ways to use monetary transfers for eliminating envy. Particularly, one can apply a balanced vector of transfer payments, or charge each agent a positive amount, or --- contrarily --- give each agent a positive amount as a ``subsidy''.
In each model, one can aim to minimize the amount of payments used; this aim translates into different optimization objectives in each setting.
This note compares the various models, and the relations between upper and lower bounds for these objectives.
\end{abstract}

\section{Introduction}
Fair allocation of indivisible items is a hard problem, but it is made easy if monetary transfers are allowed. Consider for example a setting with one item and two agents. Without money there is no fair allocation, but with money there is a simple solution: give the item to the agent who assigns to it a higher value, and let this agent pay half his value to the other agent. 
For example, if Alice values the item at $20$ and George values it at $30$, then George receives the item and pays $15$ to Alice. Both agents believe that their bundle is at least as high as the other agent's bundle, that is, the allocation is \emph{envy-free}.

But even in this simple example, there are many other envy-free allocations. For example, George can transfer to Alice any amount  in $[10,15]$. 
Another option is to charge George any price in $[20,30]$, and donate it to charity rather than giving it to Alice. A third option is to compensate Alice by paying her a subsidy in $[20,30]$.
This raises the question of which, of these uncountably-many allocations, should we choose?

This note presents several plausible optimization goals, and compares the upper and lower bounds that can be derived in each model.

\section{Related Work}
The literature on fair allocation with money is very rich and diverse. Due to time constraints, we do not survey it in this first draft. We refer the reader to our recent paper \citep{elmalem2025whoever} for a detailed literature survey.

\section{Notation}
There is a set $N$ of $n$ agents. Each agent $i$ has a \emph{weight} $w_i$ denoting its entitlement (for example, it can denote the amount of shares agent $i$ holds in the asset to be divided). 
We denote $W := \sum_{i=1}^n w_i$.

There is a set $M$ of items. each agent $i$ has an additive valuation function $v_i$ over the items.
~
The goal is to compute an allocation $(A,\mathbf{p})$, such that
\begin{itemize}
    \item $A$ is a partition of the items among the agents; $A_i$ is the set of items assigned to agent $i$.
    \item $\mathbf{p}$ is a vector of monetary payments; $p_i$ is the amount paid to agent $i$.
\end{itemize}
The payment vector $\mathbf{p}$ is called \emph{balanced} if $\sum_{i=1}^n p_i = 0$.
It is called a \emph{subsidy} if $\min_{i=1}^n p_i \geq 0$. We often denote a subsidy vector by $\mathbf{s}$ instead of $\mathbf{p}$.

Given a subset $N'\subseteq N$, we denote $p_{N'} := \sum_{i\in N'}p_i$, and similarly $s_{N'} := \sum_{i\in N'}s_i$.

Each agent $i$ has a \emph{quasilinear} utility function $u_i$, defined by: $u_i(A,\mathbf{p}) = v_i(A_i) + p_i$.
Each agent expects to receive at least the same amount  of utility-per-share as any other agent. Formally, given an allocation $(A,\mathbf{p})$, we say that agent $i$ is \emph{weighted envy-free (WEF)} towards agent $j$ if 
\begin{align*}
    (v_i(A_i) + p_i) / w_i \geq (v_i(A_j) + p_j) / w_j.
\end{align*}
The allocation $(A,\mathbf{p})$ is called WEF if given that allocation, every agent $i$ is WEF towards any other agent $j$.

\section{Comparison between different objectives}
\label{sec:models}
The recent literature focuses on attaining fairness via \emph{subsidies} --- a vector $\mathbf{s}$ with $s_i\geq 0$ for all $i \in N$. 
Two objectives in this model are upper-bounding the subsidy per agent:
\begin{align}
\tag{P1}
& \text{Find an upper bound on } \displaystyle\max_{i\in N} s_i
\\
\notag
& \text{where } (A,\mathbf{s}) \text{ is a WEF allocation and } s_i\geq 0 \text{ for all } i\in N.
\end{align}
and upper-bounding the total subsidy:
\begin{align}
\tag{P2}
& \text{Find an upper bound on } \sum_{i\in N} s_i
\\
\notag
& \text{where } (A,\mathbf{s}) \text{ is a WEF allocation and } s_i\geq 0 \text{ for all } i\in N.
\end{align}
In real life, fairness is usually attained using \emph{balanced payments} --- a vector $\mathbf{p}$ with $\sum_{i\in N}p_i = 0$. 
As large negative payments may require agents to take a loan, which incurs additional costs, it may be desirable to upper-bound the largest negative payment:
\begin{align}
\tag{P3}
& \text{Find an upper bound on } \displaystyle\max_{i\in N} (-p_i)
\\
\notag
& \text{where } (A,\mathbf{p}) \text{ is a WEF allocation and } \sum_{i\in N} p_i = 0.
\end{align}
or the sum of negative payments:
\begin{align}
\tag{P4}
& \text{Find an upper bound on } \sum_{i\in N, p_i<0} -p_i
\\
\notag
& \text{where } (A,\mathbf{p}) \text{ is a WEF allocation and } \sum_{i\in N} p_i = 0.
\end{align}
Similarly, as large positive payments may be taxable, it may be desired to upper-bound the sum of positive payments; this is clearly equivalent to (P4), as in a balanced payment vector, the two sums are equal.%

In this section we prove formal relations between these four  problems. We use the following simple lemma:
\begin{Lemma}
\label{lem:slide}
Let $(A,\mathbf{s})$ be a WEF allocation.

(a) For any constant $z$, if $\mathbf{p}$ is defined by $p_i := s_i - z  w_i$, then $(A,\mathbf{p})$ is WEF too.

(b) If $\mathbf{p}$ is defined by $p_i := s_i - w_i  \frac{s_N}{W}$, then $(A,\mathbf{p})$ is both WEF and balanced.
\end{Lemma}
\begin{proof}
(a) Since we remove the same amount of ``utility per unit entitlement'' from all agents, the weighted envy levels remain the same.

(b) WEF follows from part (a) by setting $z := s_N/W$.
Balance follows from $\sum_{i\in N} p_i = \sum_{i\in N} (s_i - z w_i) = (\sum_{i\in N} s_i) - z (\sum_{i\in N} w_i)
= s_N - z W = 0$.
\end{proof}

We now show that an upper bound the subsidy per agent (P1) is the ``strongest'' bound, as it implies meaningful upper bound on each of the other three problems.

\subsection{Maximum subsidy vs. Total subsidy}

We first show that a bound for (P1) implies a bound for (P2).
\begin{Lemma}
\label{lem:p1-implies-p2}
Let $(A,\mathbf{s})$ be a WEF allocation with subsidy.

(a) If $s_i\leq T$ for all $i\in N$, then there exists a subsidy vector $\mathbf{s}'$ such that $(A,\mathbf{s}')$ is WEF and $\sum_{i\in N}s'_i\leq (n-1)T$. The factor $(n-1)$ is tight.

(b) If $s_i\leq w_i   T$ for all $i\in N$, then there exists a subsidy vector $\mathbf{s}'$ such that $(A,\mathbf{s}')$ is WEF and $\sum_{i\in N}s'_i\leq (W-w_{1})   T$.
The factor $(W-w_{1})$ is tight.
\end{Lemma}
\begin{proof}
Let $s_{min} := \displaystyle\min_{j\in N} (s_j/w_j)$,
and define $\mathbf{s'}$ by: $s'_i := s_i - s_{min}  w_i$.
Then $(A,\mathbf{s}')$ is still WEF by Lemma \ref{lem:slide}.
Moreover, by definition of minimum, $\mathbf{s}'$ is still a subsidy vector, and there is at least one $i$ for which $s'_i=0$. The inequalities on the $s_i$ remain valid for $s'_i$ --- since $s'_i\leq s_i$.

In (a), $\sum_{i\in N}s'_i$ is at most the sum of $n-1$ positive elements, each of which is at most $T$. So the sum is at most $(n-1)T$.

The factor $(n-1)$ is tight even with equal entitlements: suppose there is one item with value $T$ to all agents, which is given to agent $1$. In any WEF allocation, all other $n-1$ agents receive subsidy $T$. Therefore, the maximum subsidy per agent is $T$  and the sum of subsidies is $(n-1)T$.

In (b), $\sum_{i\in N}s'_i$ is at most $T$ times the sum of the $n-1$ largest elements in $\mathbf{w}$, which is $W - w_{1}$. 
So the sum is at most $(W-w_{1})T$.

The factor $(W-w_{1})T$ is tight for any entitlement vector.
Suppose there is one item with value $w_{1} T$, which is given to the agent with the smallest entitlement ($w_{1}$).
In any WEF allocation, every agent $2,\ldots, n$ must receive subsidy at least $w_i T$. So $s_i \leq w_i T$, whereas the sum of subsidies is $(W-w_{1})T$.

\end{proof}

\begin{Remark}
\label{lem:p2-does-not-imply-p1}
The opposite direction of Lemma \ref{lem:p1-implies-p2} does not hold even with equal entitlements. 
Suppose there are $n-1$ items with value $(n-1)T$, and consider the WEF allocation that gives one item to each agent $1,\ldots, n-1$, and a subsidy of $(n-1)T$ to agent $n$. The sum of subsidies is $(n-1)T$, but the maximum subsidy in any WEF allocation is at least $(n-1)T$, rather than $T$.
In this sense, (P1) is strictly stronger than (P2).
\end{Remark}

\subsection{Total subsidy vs. Maximum negative payment}
The next lemma shows that a bound for (P2) implies a bound for (P3) and vice-versa.
\begin{Lemma}
\label{lem:p2-implies-p3}
(a) Let $(A,\mathbf{s})$ be a WEF allocation with subsidy. 
If $\sum_{i\in N} s_i\leq T$, then there exists a balanced payment vector $\mathbf{p}$ such that $(A,\mathbf{p})$ is WEF and $(-p_i) \leq \frac{w_i}{W}  T$ for all $i\in N$.

(b) Conversely, let $(A,\mathbf{p})$ be a WEF allocation with a balanced payment vector. 
Let $T := \displaystyle \max_i \frac{W}{w_i}(-p_i)$, so that 
 $(-p_i) \leq \frac{w_i}{W} T$ for all $i\in N$. Then there exists a subsidy vector $\mathbf{s}$ such that $(A,\mathbf{s})$ is WEF and $\sum_{i\in N} s_i\leq T$.
\end{Lemma}
\begin{proof}
(a) Define $\mathbf{p}$ by: $p_i := s_i - \frac{w_i}{W}  s_N$. Then $(A,\mathbf{p})$ is WEF and balanced by Lemma \ref{lem:slide}.
Since $s_i\geq 0$ for all $i\in N$, $p_i \geq  - \frac{w_i}{W}  s_N$, which is at least $- \frac{w_i}{W}  T$ by the lemma assumption, so $(-p_i) \leq \frac{w_i}{W}  T$ for all $i\in N$.

(b) Define $\mathbf{s}$ by: $s_i := p_i + \frac{w_i}{W}  T$.
Then $(A,\mathbf{s})$ is WEF by Lemma \ref{lem:slide}, and $\mathbf{s}$ is a subsidy vector by the lemma assumption.
Since $\sum_{i\in N} p_i = 0$,
the new sum $\sum_{i\in N} s_i = \sum_{i\in N} \frac{w_i}{W}  T = T$.
\end{proof}

\subsection{Maximum subsidy vs. Maximum negative payment}
We can combine Lemmas \ref{lem:p1-implies-p2} and \ref{lem:p2-implies-p3} to get a direct implication from (P1) to (P3):
\begin{Lemma}
\label{lem:p1-implies-p3}
Let $(A,\mathbf{s})$ be a WEF allocation with subsidy.

(a) If $s_i\leq T$ for all $i\in N$, then there exists a balanced payment vector $\mathbf{p}$ such that $(A,\mathbf{p})$ is WEF and $(-p_i) \leq \frac{w_i}{W}(n-1)T$. 

(b) If $s_i\leq w_i   T$ for all $i\in N$, 
then there exists a balanced payment vector $\mathbf{p}$ such that $(A,\mathbf{p})$ is WEF and $(-p_i) \leq w_i(1-\frac{w_{1}}{W})   T$.
\end{Lemma}

\begin{Remark}
When the entitlements are all equal to $1$, the upper bound in both cases becomes  $(-p_i) \leq \frac{n-1}{n}T$. This bound is tight: if $n-1$ agents get a subsidy of $T$ and the $n$-th agent gets no subsidy, then in the balanced payment vector, agent $n$ will have to pay $\frac{n-1}{n}T$.
\end{Remark}

\begin{Remark}
The opposite direction of Lemma \ref{lem:p1-implies-p3} does not hold even with equal entitlements. 
In the example of Remark \ref{lem:p2-does-not-imply-p1},
in a balanced payment vector, agents $1,\ldots,n-1$ pay $\frac{n-1}{n}T$ and agent $n$ receives 
$\frac{(n-1)^2}{n}T$, so the maximum negative payent is 
$\frac{n-1}{n}T$, but the maximum subsidy in any WEF allocation is at least $(n-1)T$, rather than $T$.
So (P1) is strictly stronger than (P3).
\end{Remark}

\subsection{Maximum subsidy vs. Total positive payments}
We now show that a bound for (P1) implies a bound for (P4). The proof is slightly less trivial than the previous one, and currently available for the special case of equal entitlements.
\begin{Lemma}
\label{lem:p1-implies-p4}
Let $(A,\mathbf{s})$ be a WEF allocation with subsidy.

Suppose all entitlements are equal to $1$.
If $s_i\leq T$ for all $i\in N$, then there exists a balanced payment vector $\mathbf{p}$ such that $(A,\mathbf{p})$ is WEF and $\sum_{i\in N, p_i<0}-p_i\leq Tn/4$
(equivalently, $\sum_{i\in N, p_i>0}p_i\leq Tn/4$). 
The factor $Tn/4$ is tight.

\end{Lemma}

\begin{proof}
Define $\mathbf{p}$ by: $p_i := s_i - \frac{w_i}{W}  s_N$. By Lemma \ref{lem:slide}, $(A,\mathbf{p})$ is WEF and balanced.

(a) Let $k$ be the number of agents with $p_i < 0$,
and denote $x := s_N/n$. Note that, when all entitlements are equal, $\frac{w_i}{W}  s_N = x$ for all $i\in N$.
So the negative payment of each such agent is at most $x$,
and the sum of negative payments is at most 
$k x$.

On the other hand, every positive payment is at most $(T-x)$, so the sum of 
positive payments is at most 
$(n-k)  (T-x)$.

The sum of positive payments and the sum of negative payments are equal, so there are two upper bounds on the same sum.
Every affine combination of two upper bounds is an upper bound too; we take the affine combination with coefficients $(T-X)/T$ and $x/T$
and get the upper bound 
\begin{align*}
&k x (T-x)/ T + (n-k) x (T-x)/T
\\
=&
n x (T-x)/ T = n   (x - x^2/T).
\end{align*}
We find the maximum of this expression by taking the derivative w.r.t. $x$: $n  (1 - 2 x / T) = 0$, which gives $x = T/2$ and an upper bound of 
$T n / 4$, as claimed.

The factor $T n/4$ is tight: suppose the entitlements are equal, and there are $n/2$ items of value $T$ to all $n$ agents. Consider the WEF allocation in which agents $1,\ldots, n/2$ get an item and agents $n/2+1,\ldots,n$ get a subsidy of $T$.
In any WEF allocation with a balanced payment vector, 
agents $1,\ldots, n/2$ pay $T/2$ and agents 
$n/2+1,\ldots,n$ receive a subsidy of $T/2$.
So the sum of negative payments is $(n/2)(T/2) = n T/4$.
\end{proof}

Extending Lemma \ref{lem:p1-implies-p4} to general entitlements remains an open problem.

\begin{Remark}
The opposite direction of Lemma \ref{lem:p1-implies-p4} does not hold even with equal entitlements. 

Suppose there are $n-1$ items with value $(n T / 4) \frac{n}{n-1}$, and consider the WEF allocation that gives each agent $1,\ldots, n-1$ one item and a negative payment of $(n T / 4)  \frac{1}{n-1}$,
and a positive payment of $n T / 4$ to agent $n$. The sum of positive payments is $n T/4$, 
but the maximum subsidy in any WEF allocation is at least $(n T / 4) \frac{n}{n-1}$, rather than $T$.
In this sense, (P1) is strictly stronger than (P4).
\end{Remark}

We do not know if there are any implications between (P2)=(P3) and (P4); this is another open problem.

\subsection{Conclusion}
When proving upper bounds, it is best to do so for (P1); but when proving lower bounds (-- impossibility results), it is best to do so for (P2)=(P3) and for (P4).

\subsection{Open question: Largest absolute payments}
With balanced payment vectors, upper-bounding the sum of positive payments is equivalent to upper-bounding the sum of negative payments, and both are equivalent to upper-bounding the sum of absolute payments $\sum_{i\in N} |p_i|$.

Upper-bounding the largest positive payment is \emph{not} equivalent to upper-bounding the largest negative payment, but the problems are symmetric and, probably, the same bounds can be derived for both. However, we do not know how to bound both largest negative payment and largest positive payment simultaneously, that is, the largest \emph{absolute} payment:
\begin{align}
\tag{P5}
& \text{Find an upper bound on } \sum_{i\in N} |p_i|
\\
\notag
& \text{where } (A,\mathbf{p}) \text{ is a WEF allocation and } \sum_{i\in N} p_i = 0.
\end{align}
The relations between (P5) and (P1),(P2),(P3),(P4) remain an open problem.

\bibliography{mybibfile}
\end{document}